\documentclass[conference]{IEEEtran}
\IEEEoverridecommandlockouts
\usepackage{cite}
\usepackage{amsmath,amssymb,amsfonts}
\usepackage{algorithmic}
\usepackage{graphicx}
\usepackage{textcomp}
\usepackage{xcolor}
\def\BibTeX{{\rm B\kern-.05em{\sc i\kern-.025em b}\kern-.08em
    T\kern-.1667em\lower.7ex\hbox{E}\kern-.125emX}}
\usepackage{amsthm}
\usepackage[acronym]{glossaries}
\usepackage[normalem]{ulem}
\usepackage{cancel}
\usepackage{float}
\usepackage{xcolor}
\usepackage{amsmath}
\usepackage{amsfonts}
\usepackage{amssymb}
\usepackage{breqn}
\usepackage{cite}
\usepackage{hyperref}
\usepackage{etoolbox}
\hypersetup{
  colorlinks=true,
  citecolor=blue,
  linkcolor=blue,
  urlcolor=blue
}
\usepackage{adjustbox}
\usepackage{mathrsfs}

\newtheorem{theorem}{Theorem}
\newtheorem{lemma}{Lemma}
\newtheorem{assumption}{Assumption}
\newtheorem{definition}{Definition}
\newtheorem{problem}{Problem}
\newtheorem{remark}{Remark}
    
\begin{document}

\title{Distributed Resilient Asymmetric Bipartite Consensus: A Data-Driven Event-Triggered Mechanism
}

\author{\IEEEauthorblockN{1\textsuperscript{st} Yi Zhang}
\IEEEauthorblockA{\textit{The Department of Electrical and Computer Engineering} \\
\textit{University of Connecticut}\\
Storrs Mansfield, USA \\
yi.2.zhang@uconn.edu}
\and
\IEEEauthorblockN{2\textsuperscript{nd} Mohamadamin Rajabinezhad}
\IEEEauthorblockA{\textit{The Department of Electrical and Computer Engineering} \\
\textit{University of Connecticut}\\
Storrs Mansfield, USA \\
Mohamadamin.rajabinezhad@uconn.edu}
\and
\IEEEauthorblockN{3\textsuperscript{rd} Shan Zuo}
\IEEEauthorblockA{\textit{The Department of Electrical and Computer Engineering} \\
\textit{University of Connecticut}\\
Storrs Mansfield, USA \\
shan.zuo@uconn.edu
}

}

\maketitle

\begin{abstract}
The problem of asymmetric bipartite consensus control is investigated within the context of nonlinear, discrete-time, networked multi-agent systems (MAS) subject to aperiodic denial-of-service (DoS) attacks. To address the challenges posed by these aperiodic DoS attacks, a data-driven event-triggered (DDET) mechanism has been developed. This mechanism is specifically designed to synchronize the states of the follower agents with the leader’s state, even in the face of aperiodic communication disruptions and data losses. Given the constraints of unavailable agents' states and data packet loss during these attacks, the DDET control framework resiliently achieves leader-follower consensus. The effectiveness of the proposed framework is validated through two numerical examples, which showcase its ability to adeptly handle the complexities arising from aperiodic DoS attacks in nonlinear MAS settings.
\end{abstract}


\section{INTRODUCTION}
\label{sec:introduction}

The topic of consensus in multi-agent systems (MAS) has garnered significant attention due to its focus on achieving a common agreement among a group of agents through neighborhood negotiations. It is well-recognized that the advantageous outcomes of achieving consensus in MAS have found extensive applications in the field of control engineering. These applications include, but are not limited to, flocking, swarming, formation control, and synchronization \cite{olfati2004consensus}.

In practical scenarios, the dynamics within MAS often display a mix of cooperative and competitive interactions, similar to the dual nature of trust and distrust in social networks or competition in economic markets. These dynamics significantly affect agent collaboration. Altafini introduced the concept of bipartite consensus for integrator MAS, using graph theory to model complex relationships among agents with both positive and negative interactions represented within a signed graph \cite{altafini2012consensus}. This model allows agents to either oppose each other with equal force or converge towards a neutral position, especially in the absence of structural balance.

Real-world scenarios often require an asymmetrical approach to bipartite consensus. For example, multilateral teleoperation systems adjust force feedback based on the differing masses of equipment \cite{chen2019rbfnn}. Similarly, rehabilitation robots for lower limbs provide customized assistance or resistance according to the muscle strength of each patient’s leg \cite{hussain2013adaptive}. These applications highlight the need for strategies tailored to the specific requirements of asymmetric consensus, which is a central theme of this research.

Developing accurate mathematical models for nonlinear MAS is often a time-consuming and challenging task, further complicated by unmodeled dynamics and a lack of robustness. These issues can lead to overly complex control designs that are impractical, especially when added network complexities further complicate matters \cite{zhang2019networked,wang2022neural}. As a solution, there is a growing trend in the field towards adopting data-driven control approaches that eliminate the need for detailed system modeling.

Furthermore, MAS are increasingly exposed to cybersecurity risks \cite{zuo2022adaptive,sun2023resilience,zuo2022resilient,zuo2023resilient,zhang2024resilient}, particularly through cyber-attacks such as aperiodic DoS attacks, due to their reliance on communication networks. These attacks disrupt communication, causing data loss, miscoordination, and severe impacts on system performance, potentially leading to complete failure \cite{cetinkaya2019overview}. The decentralized and interconnected nature of MAS complicates the application of traditional cybersecurity measures, necessitating tailored security solutions for each agent in dynamically changing networks. Aperiodic DoS attacks, by interrupting data flow, pose significant threats to the integrity and functionality of MAS, challenging the cooperative control strategies essential for system consensus and synchronization. This underscores the urgent need for robust control methods designed to maintain reliable communication and achieve system goals despite these disruptions.

Meanwhile, MAS frequently comprise multiple control loops that share computational and communication resources. Efficient allocation of these resources is a crucial aspect of MAS design. Traditional digital control methods typically operate under a time-triggered paradigm, where controllers execute periodically regardless of the system's state. This approach can lead to excessive workloads, especially when these resources could be better utilized for other tasks. Due to these limitations, there has been a renewed interest in event-triggered control \cite{ge2021dynamic}. In event-triggered control systems for MAS, the inputs are updated not at regular intervals, but in response to specific events. Recently, considerable effort has been focused on developing systematic techniques for designing event-triggering mechanisms that effectively implement stabilizing feedback controllers.

The main contributions of this paper are: 1) We enhance existing asymmetric bipartite consensus (ABC) formulations by introducing asymmetric coefficients for both follower groups, creating a more practical unified resilient ABC problem, with an associated neighborhood ABC error. 2) We develop a novel, distributed resilient data-driven event-triggered (DDET) adaptive control algorithm using a data-driven, time-varying parameter. This algorithm, distinct from current methods \cite{hu2021event,cong2021sampled}, is tailored to mitigate the impact of aperiodic DoS attacks in collaborative ABC control and operates without the need for traditional mathematical modeling of system dynamics.

\noindent\textbf{Notations:} In this paper, $\left|X^{n}\right|$ and $\left|X^{m \times n}\right|$ denote the Euclidean and 2-norms for vectors and matrices, respectively. Vectors $\mathbf{1}_N \in \mathbb{R}^N$ and $\mathbf{0}$ consist of ones and zeros, respectively. The identity matrix is represented by $I$. Elements of a matrix $A \in \mathbb{R}^{n \times n}$ are denoted by $\left[A\right]_{ij}$.

\section{Preliminaries}
\label{sec:pre}

We consider a MAS consisting of one leader and $N$ followers, where the interactions among them are represented by a signed digraph $\mathscr{G}=(\mathcal{V}, \mathscr{E}, \mathcal{A})$, where $\mathcal{V}=\{0,1,2, \ldots, N\}$ is the set of vertices, $0$ denotes the leader, and $1\cdots N $ denote the followers. $\mathscr{E} \subset \mathcal{V} \times \mathcal{V}$ denotes the set of edges, and $\mathcal{A}=[a_{ij}]\in\mathbb{R}^{N\times N}$ is the associated adjacency matrix where $a_{ij} \neq 0$ if $(j, i) \in \mathscr{E}$. The neighborhood of the agent $i$ is $\mathcal{N}_i=\{j\in\mathcal{V}:(j,i)\in\mathscr{E}\}$ and the self-edge $(i,i)$ satisfies $(i,i)\notin\mathscr{E}$. The in-degree matrix is defined as $\mathcal{D}=\operatorname{diag}\left(d_i\right)$ with $d_i=\sum\nolimits_{j \neq i}a_{ij}$. The Laplacian matrix ${\mathcal{L}}$ is defined as
$\mathcal{L} = \mathcal{D} - \mathcal{A}$. We then introduce the following definitions for a structurally balanced signed graph $\mathscr{G}$ and its associated graph matrix $L_{|\mathscr{G}|}$.

\begin{definition}
\label{def:1}
The signed graph $\mathscr{G}$ is structurally balanced if the vertex set $\mathcal{V}$ can be partitioned into $\mathcal{V}_{1}$ and $ \mathcal{V}_{2}$, with $\mathcal{V}_{1} \cup \mathcal{V}_{2}=\mathcal{V}$ and $\mathcal{V}_{1} \cap \mathcal{V}_{2}=\emptyset$, such that $a_{ij} \geqslant 0, \forall i,j \in \mathcal{V}_{\iota}$ and $a_{ij} \leqslant 0, \forall i \in \mathcal{V}_{\iota}, j \in \mathcal{V}_{3-\iota}$, where $\iota=\{1,2\}$.
\end{definition}
\begin{definition}
\label{def:2}
    For a structurally balanced signed digraph $\mathscr{G}$, we define $W=\mathrm{diag}\left(\delta_i\right)$ and $S=\mathrm{diag}\left(s_i\right)$, where
\begin{equation}
\label{eq1}
\delta_i=\left\{\begin{aligned}
& 1, &i \in \mathcal{V}_{1}
\\&-1,&i \in \mathcal{V}_{2}
\end{aligned}\right. \text{ and } s_i=\left\{\begin{aligned}
& m, &i \in \mathcal{V}_{1}
\\&n,&i \in \mathcal{V}_{2}
\end{aligned}\right.\end{equation}

Then an associated graph matrix $L_{|\mathscr{G}|}$ is defined as
\begin{equation}
\label{eq2}
L_{|\mathscr{G}|}=\mathrm{diag}\left(S^{-1}W\mathcal{A}WS\mathbf{1}_N\right)-\mathcal{A} 
\end{equation}
\end{definition}

\section{System Modeling and Problem Formulation}
\label{sec:prb}

Consider a discrete-time nonlinear system for $N$ followers, each governed by
\begin{equation}
\label{eq3}
    y_i(k+1) = f_i\left(y_i(k), u_i(k)\right), \quad i \in \mathcal{V}
\end{equation}
where $y_i(k) \in \mathbb{R}$ and $u_i(k) \in \mathbb{R}$ denote the output and input of follower $i$ at time step $k$. Here, $f_i(\cdot)$ is an unknown smooth nonlinear function, and $\mathcal{V}$ represents the set of all followers.

\begin{assumption}\label{ass:1}
Partial derivative $\frac{\partial f_i(\cdot, \cdot)}{\partial u_i(k)}$ is continuous.
\end{assumption} 
To achieve the leader-follower consensus, we have following assumption for the MAS \eqref{eq3}.
\begin{assumption}\label{ass:2}
The signed digraph $\mathscr{G}=(\mathcal{V}, \mathscr{E}, \mathcal{A})$ is structurally balanced and contains a spanning tree with the leader as the root.
\end{assumption}

\begin{remark}
\label{rem:1}
Assumption \ref{ass:1} serves as a broad constraint for the nonlinear system.
To achieve any leader-follower consensus control, Assumption \ref{ass:2} is a necessary
and sufficient condition \cite{lewis2013cooperative}. Besides, unlike the strategies in \cite{guo2018asymmetric} and \cite{liang2021finite} that require a strongly connected communication directed graph, our assumption broadens its applicability by accommodating a weakly connected digraph. 
\end{remark}

\begin{assumption}[\cite{hou2010novel}]\label{ass:3}
The system \eqref{eq3} satisfies the generalized Lipschitz condition, i.e., for time instants $k+1, k \geqslant 0$ and {$u_i(k+1) \neq u_i(k)$}, there exists a positive constant $b_i^y$ such that $\left|\Delta y_i(k+1)\right|=b_i^y(k) \left|\Delta u_i(k)\right|$, where $\Delta y_i(k+1)=y_i(k+1)-y_i(k)$ and $\Delta u_i(k)=u_i(k)-u_i(k-1)$.
\end{assumption}

\begin{remark}
\label{rem:2}
Assumption \ref{ass:3} suggests that the system output increment is constrained by input increment, which is
satisfied in real-world systems.     
\end{remark}
\begin{lemma}[\cite{hou2010novel}]
\label{lem:1}
For nonlinear systems \eqref{eq3} satisfying the Assumptions \ref{ass:1} and \ref{ass:3}, there exists a time-varying parameter called pseudo partial derivative (PPD) $\mathscr{M}_{i}^y(k) \in \mathbb{R}$ such that the systems \eqref{eq3} can be equivalently transformed into the following compact-form dynamic linearization (CFDL) data model
\begin{equation}
\label{eq4}
\Delta y_i(k+1)=\mathscr{M}_i^y(k) \Delta u_i(k)    
\end{equation}
where $\left|\mathscr{M}_i^y(k)\right| \leqslant b_i^y$ is bounded with $\forall k$ and $i \in \mathcal{V}$.   
\end{lemma}

The unified ABC objective is proposed as follows.
\begin{definition}
\label{def:3}
Given a reference value $y_{d}(k)$ issued by the leader, the unified ABC objective is to achieve
\begin{equation}
\label{eq5}
\lim_{k \to \infty} y_i(k)=\left\{\begin{aligned}
&my_{d}(k),  &i \in \mathcal{V}_{1} \\
&-n y_{d}(k), &i \in \mathcal{V}_{2}
\end{aligned},\quad \mathcal{V}_{1} \cup \mathcal{V}_{2}=\mathcal{V}\right.
\end{equation}
where $m$ and $n$ are influence coefficients, which are positive.
\end{definition}

\begin{remark}
\label{rem:3}
Our model \eqref{eq5} introduces distinct asymmetric coefficients $m$ and $n$, unlike traditional single-coefficient ABC models \cite{guo2018asymmetric,liang2021finite}. It is effectively applied in robotic lower limb rehabilitation to apply tailored torques \cite{hussain2013adaptive}. When $m$ and $n$ equal $1$, it simplifies to standard bipartite consensus objective.
\end{remark}

To achieve the unified ABC objective \eqref{eq5}, we introduce the following neighbourhood ABC error
\begin{equation}
\label{eq6}
\begin{aligned}
&e_{y_i}(k)=
\\&\left\{\begin{aligned} 
&\left(\begin{aligned}&\sum_{j \in \mathcal{V}_{1}} \left(a_{ij}y_{j}(k)-\left|a_{ij}\right| y_i(k)\right)
\\&+\sum_{o \in \mathcal{V}_{2}} \left(a_{io}y_{o}(k)-m^{-1}n \left|a_{io}\right| y_i(k)\right)
\\&+\left(g_iy_{d}(k)-m^{-1}\left|g_i\right| y_i(k)\right)\end{aligned}\right),& i \in \mathcal{V}_{1}
\\&\left(\begin{aligned}&\sum_{j \in \mathcal{V}_{1}} \left(a_{ij}y_{j}(k)-n^{-1} m\left|a_{ij}\right| y_i(k)\right) 
\\&+\sum_{o \in \mathcal{V}_{2}} \left(a_{io}y_{o}(k)-\left|a_{io}\right| y_i(k)\right)
\\&+\left(g_iy_{d}(k)-n^{-1}\left|g_i\right| y_i(k)\right)\end{aligned}\right),&i \in \mathcal{V}_{2}
\end{aligned}\right.
\end{aligned}
\end{equation}
The pinning gain $g_i$ from the leader to follower $i$ is determined by the availability of $y_{d}(k)$ to agent $i$. It is defined as follows
\[
g_i = \begin{cases} 
1, & i \in \mathcal{V}_{1} \\
-1, & i \in \mathcal{V}_{2} \\
0, & \text{otherwise}
\end{cases}
\]

The following technical result is needed.

\begin{lemma}[\cite{zhang2024resilient}]
\label{lem:3}
Under Assumption \ref{ass:2}, consider MAS \eqref{eq3}, the unified ABC objective \eqref{eq5} is achieved if and only if $\lim\limits_{k\to\infty}\mathbf{e}_y(k) =\mathbf{0}$. 
\end{lemma}

Besides, the following assumtion and the induced lemma is necessary for the proof in the later sections.
\begin{assumption}
[\cite{ma2021distributed}]\label{ass:4}
For any $k$, $\frac{\left\|\Delta u_{j}(k)\right\|}{\left\|\Delta u_i(k)\right\|}<\sigma_{ij}$, $i, j \in \mathcal{V}$, where $\sigma_{ij}$ is a positive constant.
\end{assumption}
\begin{lemma}[\cite{ma2021distributed}]
\label{lem:4}
Given Assumptions \ref{ass:1}, \ref{ass:2}, \ref{ass:3}, and \ref{ass:4}, and assume condition $\left|u_i(k)\right|>\epsilon_i$ holds, where $\epsilon_i$ is a positive constant, then there exists a d pseudo-partial derivative (PPD) $\mathscr{M}_i(k)$ such
that \eqref{eq6} can be equivalently transformed into the following CFDL data model
\begin{equation}
\label{eq8}
\Delta e_{y_i}(k+1)=\mathscr{M}_i(k) \Delta u_i(k)
\end{equation}where $\Delta e_{y_i}(k+1)=e_{y_i}(k+1)-e_{y_i}(k)$ and $\mathscr{M}_i(k) \in \mathbb{R}$ is bounded, i.e., $\left|\mathscr{M}_i(k)\right| \leqslant b_i$, where $b_i$ is a 
positive constant.
\end{lemma}

We model the aperiodic DoS attack we aim to address as follows. The $\ell$th DoS attack interval is expressed as

\begin{equation}
\label{eq9}
K_{\ell}=\left[k_{\ell}, k_{\ell}+\theta_{\ell}\right)
\end{equation}
where $K_{\ell}$ determines when the transmission is interrupted, $k_{\ell}$ is the moment that the attacker strikes, $\theta_{\ell}$ represents the duration of the $\ell$th DoS attack, and $k_{\ell+1} > k_{\ell} + \theta_{\ell}$.

Furthermore, DoS attacks aim to block the transmission of data packets, consequently reducing system performance and leading to data packet dropouts. Since $e_{y_i}(k)$ contains all the exchange information $y_i$ and $y_d$, the data packets received by the controller during DoS attacks are represented as
\begin{equation}
\label{eq10}
\bar{e}_{y_i}(k)=h_i(k) e_{y_i}(k)
\end{equation}
where $
h_i(k)=\left\{\begin{aligned}
&1,&& k \in\left[k_{\ell-1}+\theta_{\ell-1}, k_{\ell}\right) 
\\&0,&& k \in\left[ k_{\ell}, k_{\ell}+\theta_{\ell}\right)
\end{aligned}\right.
$.

\begin{assumption}
[DoS attack frequency\cite{de2015input}]
\label{ass:5}
For any $\left[k_{0}, k\right] \subset$ $[0, \infty)$, there must be positive constants $\kappa_{a}$ and $K_{\infty}$ function $f_{a}\left(k_{0}, k\right)$ such that

\begin{equation}
\label{eq12}
n_{a}\left(k_{0}, k\right) \leq \kappa_{a}+f_{a}\left(k_{0}, k\right),
\end{equation}
where $n_{a}\left(k_{0}, k\right)$ is the amount of DoS attack occurring over $\left[k_{0}, k\right]$.    
\end{assumption}
\begin{assumption}[DoS attack duration\cite{de2015input}]
\label{ass:6}
For any $\left[k_{0}, k\right] \subset$ $[0, \infty)$, there must be positive constants $\zeta_{a}$ and $K_{\infty}$ function $f_{\Xi}\left(k_{0}, k\right)$ such that

\begin{equation}
\label{eq13}
\Xi_{a}\left(k_{0}, k\right) \leq \zeta_{a}+f_{\Xi}\left(k_{0}, k\right) .
\end{equation}
\end{assumption}
\begin{remark}
\label{remark4}
Assumptions \ref{ass:5} and \ref{ass:6} avoid making assumptions about the probability distribution of the attacker's strategy which is different from \cite{zhang2023resilient,zhang2024resilient}. The frequency of DoS attacks is limited by an upper bound, as indicated by Assumption \ref{ass:5}. Similarly, Assumption \ref{ass:6} limits the total duration of DoS attacks. References for Assumptions \ref{ass:5} and \ref{ass:6} are provided in \cite{de2015input,yang2020distributed}.
\end{remark}

Now we introduce the following resilient ABC problem.
\begin{problem}[Resilient ABC Problem]
\label{prb:1}
Given Assumptions \ref{ass:1}, \ref{ass:2}, and \ref{ass:3}, and considering the MAS \eqref{eq3} under aperiodic DoS attacks, the objective of the ABC problem is to develop an DDET algorithm. This algorithm aims to maintain the ABC error $e_i(k)=y_d(k)-s_i^{-1}\delta_iy_i(k)$ within a specified boundary. Specifically, there exists a positive constant $b_i$ such that $\left|e_i(k)\right| \leqslant b_i$, where the sequence of event-triggered instants $k_{t}^{i}$ is determined by the following criterion
\begin{equation}
\label{eq14}    k_{t}^{i}=\inf \left\{k \geq k_{t-1}^{i} \bigg| g_i\left(\Delta_{i}(k), \tilde{e}_{y_i}(k)\right)<0, t \in \mathbb{N}\right\}
\end{equation}
where $g_i$ is a event-triggered function for agent $i$ and $k_{0}^{i}=1$.
\end{problem}

\section{DDET Algorithm Design}

In this section, the DDET algorithm is designed. For each agent $i$, to counteract aperiodic DoS attacks, reduce continuous communication, and save network resources, we introduce the event-triggered neighbourhood ABC error $\tilde{e}_{y_i}(k)$ as follows
\begin{flalign*}
\label{eq15}
\begin{aligned}
&\tilde{e}_{y_i}\left(k\right)=
\\&\left\{\begin{aligned}
&\left(\begin{aligned}&\sum_{j \in \mathcal{V}_{1}} \left(a_{ij}y_{j}\left(k_{t}^{j}\right)-\left|a_{ij}\right| y_i\left(k_{t}^{i}\right)\right)
\\&+\sum_{o \in \mathcal{V}_{2}} \left(a_{io}y_{o}\left(k_{t}^{o}\right)-m^{-1}n \left|a_{io}\right| y_i\left(k_{t}^{i}\right)\right)
\\&+\left(g_iy_{d}(k)-m^{-1}\left|g_i\right| y_i\left(k_{t}^{i}\right)\right)\end{aligned}\right),&i \in \mathcal{V}_{1}
\\&\left(\begin{aligned}&\sum_{j \in \mathcal{V}_{1}} \left(a_{ij}y_{j}\left(k_{t}^{j}\right)-n^{-1} m\left|a_{ij}\right| y_i\left(k_{t}^{i}\right)\right) 
\\&+\sum_{o \in \mathcal{V}_{2}} \left(a_{io}y_{o}\left(k_{t}^{o}\right)-\left|a_{io}\right| y_i\left(k_{t}^{i}\right)\right)
\\&+\left(g_iy_{d}(k)-n^{-1}\left|g_i\right| y_i\left(k_{t}^{i}\right)\right)\end{aligned}\right),& i \in \mathcal{V}_{2}
\end{aligned}\right.
\end{aligned}
&&\raisetag{5\baselineskip}
\end{flalign*}
where $\tilde{e}_{y_i}(k)$ represents the consensus error based on the triggered outputs and $k_{t}^{i}$ is the event-triggered time instant.

The distributed event-triggered condition is designed as follows
\begin{equation}
\label{eq16}
g\left(\tilde{e}_{y_i}(k), \Delta_{i}(k)\right)=\left|\tilde{e}_{y_i}(k-1)\right|-\theta(k-1)\left|\Delta_{i}(k-1)\right|<0
\end{equation}
which is used to determine the sequence of event-triggered instants $k_{t}^{i}$ for each agent $i$. The term $\Delta_{i}(k) = s_{i}^{-1} \delta_{i}(y_{i}(k) - y_{i}(k_{t}^{i}))$ represents the event-triggered output error, while $\theta(k)$ is time-varing parameter which will be determined later.

The DDET protocol is designed as
\begin{align}
&\begin{aligned}
\label{eq17}
&\hat{\mathscr{M}}_i(k)=
\\&\left\{\begin{aligned}&\begin{aligned}
&\hat{\mathscr{M}}_i(k_{t-1}^{i})
-\frac{\eta_{1i}\hat{\mathscr{M}}_i(k_{t-1}^{i})\Delta u_i^2(k_{t-1}^{i}) }{\Delta u_i^2(k_{t-1}^{i})+\mu_i}
\\&+\frac{\eta_{1i}\Delta u_i(k_{t-1}^{i})}{\Delta u_i^2(k_{t-1}^{i})+\mu_i}
\left(\bar{e}_{y_i}(k_{t}^{i})-\bar{e}_{y_i}(k_{t}^{i}-1)\right)
\end{aligned},&&k = k_{t}^{i}
\\
&\hat{\mathscr{M}}_i(k_{t-1}^{i}), && k \in \left(k_{t-1}^{i}, k_{t}^{i}\right)\end{aligned}\right.
\end{aligned}&&\raisetag{4\baselineskip}
\\\label{eq18}&\begin{aligned}
&\hat{\mathscr{M}}_{i}(k)=\hat{\mathscr{M}}_{i}(1),\text { if }\left|\hat{\mathscr{M}}_{i}(k)\right|<\gamma \text { or }\left|\Delta u_{i}(k)\right|<\gamma  
\\&\text{ or } \operatorname{sgn}\left(\hat{\mathscr{M}}_{i}(k)\right) \neq \operatorname{sgn}\left(\hat{\mathscr{M}}_{i}(1)\right)
\end{aligned}
\\\label{eq19}&\begin{aligned}
&u_i(k) =
\\&\left\{\begin{aligned}&u_i(k_{t-1}^{i}) +\frac{\eta_{2i}\hat{\mathscr{M}}_i(k_{t-1}^{i})}{\hat{\mathscr{M}}_i(k_{t-1}^{i})^2+\varpi_i}h_i(k_{t}^{i})\tilde{e}_{y_i}\left(k\right),&&k = k_{t}^{i}
\\&u_i(k_{t-1}^{i}),&&k \in \left(k_{t-1}^{i}, k_{t}^{i}\right)
\end{aligned}\right.
\end{aligned}&&\raisetag{2.5\baselineskip}
\end{align}
Note that when a DoS attack occurs at the event-triggered time, that is, when $h_i(k_{t}^{i})=0$, the control input will not be updated during the DoS attack. Furthermore, if the agent's event is valid at $k=k_{t}^{i}$, the output $\tilde{e}_{y_i}\left(k_{t}^{i}\right)=\tilde{e}_{y_i}(k)$ becomes the triggered output of the update controller and remains unchanged until the next event-triggered instant $k_{t+1}^{i}$. Meanwhile, $\tilde{e}_{y_i}\left(k_{t}^{i}\right)$ is released into the network, which can reduce communication between agents.
\begin{theorem}\label{thm:1}
Consider a MAS \eqref{eq4}, with  Assumptions \ref{ass:1}, \ref{ass:2}, \ref{ass:3}, and \ref{ass:4} holds, under the event-triggered condition \eqref{eq16}, there exists a time-varying parameter $\theta(k)$ such that 
\begin{equation}
\label{eq20}
\theta(k)=\frac{2\sigma_{\max}\left(\varPsi^{\mathrm{T}}P(k)\right)}{2\sigma_{\min}\left(P(k)\right)+\sigma_\max \left(P(k)^2\right)}
\end{equation}
where $P(k) = - \mathrm{diag}\left(\frac{\eta_{2i}\mathscr{M}_i(k)\hat{\mathscr{M}}_i(k)}{\hat{\mathscr{M}}_i(k)^2+\varpi_i}\right)$, and $\varPsi=L_{|\mathscr{G}|}WS+\mathcal{G}$. The local ABC error $e_i(k)$ for the system under aperiodic DoS attacks is bounded, that is the ABC problem is solved.
\end{theorem}
\begin{proof}
    Due to the page limit, the detailed proof is omitted. The simplified proof is given. Firstly, $\tilde{e}_{y_i}(k)$ is expressed in a compact form
\begin{equation}
\label{eq31}
\begin{aligned}
\tilde{e}_{y_i}(k)=&\sum_{j \in \mathcal{N}_i} \left(a_{ij}y_{j}(k_{t}^{j})-\delta_is_i^{-1}a_{ij}s_j\delta_j y_i(k_{t}^{i})
\right)
\\&+g_i\left(y_{d}(k)-\delta_is_i^{-1} y_i(k_{t}^{i})\right), i \in \mathcal{V}
\end{aligned}
\end{equation}

Define a Lyapunov function candidate $V(k)=\mathbf{e}_y(k)^{\mathrm{T}}\mathbf{e}_y(k)$, where $\mathbf{e}_y(k)=\mathrm{col}[e_{y_i}]$. Then we have
\begin{flalign}
\label{eq38}
\begin{aligned}
\Delta V(&k+1)=\mathbf{e}_y(k+1)^{\mathrm{T}}\mathbf{e}_y(k+1)-\mathbf{e}_y(k)^{\mathrm{T}}\mathbf{e}_y(k)
\\=&\left(\mathbf{e}_y(k)-P(k)\tilde{\mathbf{e}}_y(k) \right)^{\mathrm{T}}\left(\mathbf{e}_y(k)-P(k)\tilde{\mathbf{e}}_y(k) \right)
\\&-\mathbf{e}_y(k)^{\mathrm{T}}\mathbf{e}_y(k)
\\=&-2\mathbf{e}_y(k)^{\mathrm{T}}P(k)\tilde{\mathbf{e}}_y(k)
+\tilde{\mathbf{e}}_y(k)^{\mathrm{T}}P(k)^2\tilde{\mathbf{e}}_y(k)
\\=&-2\tilde{\mathbf{e}}_y(k)^{\mathrm{T}}P(k)\tilde{\mathbf{e}}_y(k)+2\Delta(k)^{\mathrm{T}}\varPsi^{\mathrm{T}} P(k)\tilde{\mathbf{e}}_y(k)
\\\leqslant&-2\sigma_{\min}\left(P(k)\right)\left\|\tilde{\mathbf{e}}_y(k)\right\|^2+\sigma_\max \left(P(k)^2\right)\left\|\tilde{\mathbf{e}}_y(k)\right\|^2
\\&+2\sigma_{\max}\left(\varPsi^{\mathrm{T}}P(k)\right)\left\|\Delta(k)\right\|\left\|\tilde{\mathbf{e}}_y(k)\right\|
\end{aligned}&&\raisetag{4\baselineskip}
\end{flalign}

Denote $\theta(k) = \frac{2\sigma_{\max}\left(\varPsi^{\mathrm{T}}P(k)\right)}{2\sigma_{\min}\left(P(k)\right)+\sigma_\max \left(P(k)^2\right)}$. There is a sufficient condition for $\Delta V\leqslant 0$ is
$\theta (k)\left|\Delta_i(k)\right|-\left|\tilde{e}_{y_i}(k)\right|\leqslant 0$. Since when $\Delta V\leqslant 0$, there no need any update action. Hence, we get a sufficient event-triggered condition \eqref{eq16} with time-varying parameter \eqref{eq20}. 
\end{proof}

\section{Numerical Example}
\label{sec:num}
In this section, a structurally balanced communication graph for $8$ agents is considered, as shown in Fig.\ref{Fig.1}. We denote $\mathcal{V}_{1}=\{1,2,4,5,6,8\}$ and $\mathcal{V}_{2}=\{3,7\}$. 
Then, two numerical examples are introduced to demonstrate the proposed DDET's effectiveness and resilience. First, we examine its performance for nonlinear MAS under aperiodic DoS attacks with constant consensus reference leader. Then, we explore its application to MAS under aperiodic DoS with time-varying consensus reference leader.

Agent $1$ is the only root in the graph, which has a directed path to the remaining seven agents. However, directed paths among these $7$ agents do not all exist. Hence, the given digraph is not strongly connected compared to \cite{guo2018asymmetric} and \cite{liang2021finite} but has a spanning tree. The influence coefficients are $m=3$ and $n=4$. To run the proposed DDET protocol, the initial values are set as follows: $y_{i}(1)=0.5$, $u_{i}(1)=0, \hat{\mathscr{M}}_{i}(1)=1$, $\theta(0) =0$, and $P(0) =0$. The other necessary parameters are chosen as $\eta_{1i}=\eta_{2i}=0.1$, $ \mu_i=1$, and $ \varpi_i=-0.1$.

\begin{figure}[ht]
\label{Fig.1}
\centering
\includegraphics[width=0.35\textwidth]{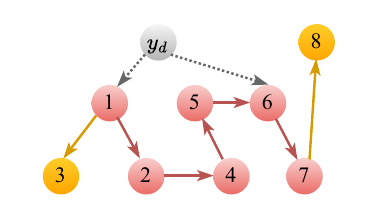}
\caption{The structurally balanced, signed, directed communication graph for networked agents.}
\end{figure}
\subsection{DDET under Aperiodic DoS with Constant Leader}

In this numerical example, the models of $8$ agents are given to produce the nonlinear data flow as follows 
\begin{equation*}
\begin{aligned}
&y_{i}(k+1) =
\\&\left\{\begin{aligned}
&\frac{y_{i}(k) u_{i}(k)}{1 + y_{i}(k)^{w_{i1}}} 
+ w_{i2} u_{i}(k)  \tanh(u_{i}(k)),i \in \{1,2\} &&   
\\& \frac{y_{i}(k) u_{i}(k)}{1 + y_{i}(k)^{w_{i1}}}+ w_{i2} u_{i}(k)  \sin(u_{i}(k)), i \in \{3,4\}&&
    \\& \frac{y_{i}(k) u_{i}(k)}{1 + y_{i}(k)^{w_{i1}} + u_{i}(k)} 
    + w_{i2} u_{i}(k)  \cos(u_{i}(k)), i \in \{5,6\}
    \\& \frac{y_{i}(k) u_{i}(k)}{1 + y_{i}(k)^{w_{i1}} + u_{i}(k)} + w_{i2} u_{i}(k)  \log(1 + |u_{i}(k)|), i \in \{7,8\}
\end{aligned}\right.
\\&\begin{aligned}
&w_{11} = 4,  w_{12} = 2, 
w_{21} = 3,  w_{22} = 3,
\\&w_{31} = 4,  w_{32} = 4,
w_{41} = 3,  w_{42} = 5, 
\\&w_{51} = 2,  w_{52} = 2,
w_{61} = 4,  w_{62} = 0.4, 
\\&w_{71} = 2,  w_{72} = 2,
w_{81} = 3,  w_{82} = 0.5.    
\end{aligned}
\end{aligned}
\end{equation*}

The desired reference consensus signal is
\begin{equation*}
y_d(k)=\left\{\begin{aligned}
&3, &&k \in [0,900)
\\&2, &&k \in [900, 1700)
\\&1, &&k \in [1700, 2500)
\end{aligned}\right.   
\end{equation*}

\begin{figure}[ht]
\centering
\includegraphics[width=0.5\textwidth]{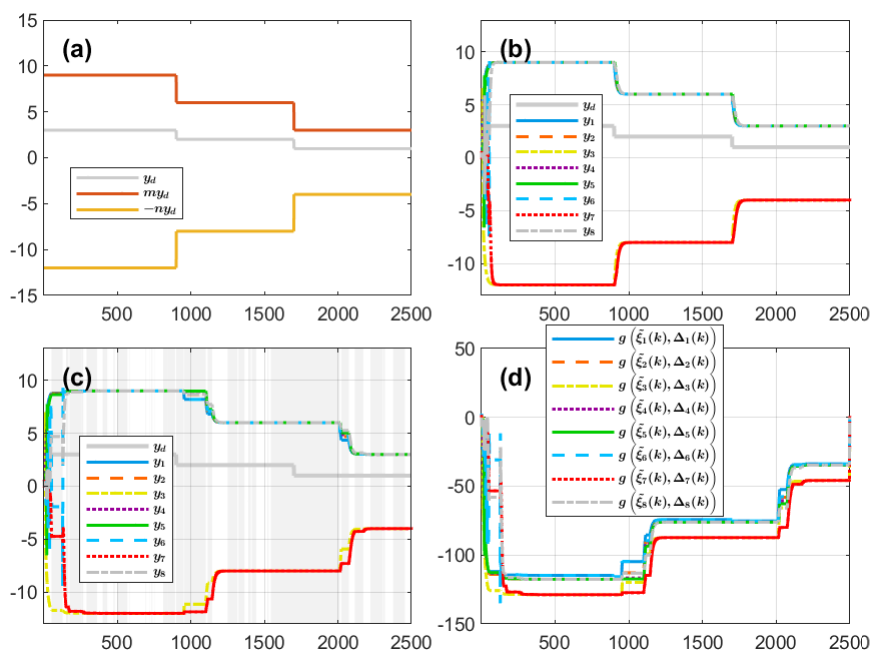}
\caption{(a) The grey line is reference leader $y_d$, the red line is the desired trajectory of $\mathcal{V}_1$, and the yellow line is the desired trajectory of $\mathcal{V}_2$. (b) Output trajectories under the DDET protocols without aperiodic DoS attacks: The gray solid line represents the constant leader reference signal $y_d(k)$. (c) Output trajectories $y_i(k)$ under Aperiodic DoS attacks while using the DDET protocols: The gray solid line represents the constant leader reference signal $y_d(k)$. (d) The event-triggered function value $g_i\left(\Delta_{i}(k), \tilde{e}_{y_i}(k)\right)$ based on the developed data-driven parameter while under aperiodic DoS attacks.}
\label{Figure_1}
\end{figure}

The outputs of the agents are profiled in Fig. \ref{Figure_1}. It is clear that over the time interval \(k \in [0,900)\), the outputs \(y_{i}\) for \(i \in \mathcal{V}_{1}\) track within a tiny boundary of 9, while \(y_{i}\) for \(i \in \mathcal{V}_{2}\) track within a small boundary of -12, with respect to \(y_{d}(k)=3\). For the interval \(k \in [900, 1700)\), \(y_{i}\) for \(i \in \mathcal{V}_{1}\) track within a tiny boundary of 6, and \(y_{i}\) for \(i \in \mathcal{V}_{2}\) track within a small boundary of -8, with respect to \(y_{d}(k)=2\). Additionally, over the time interval \(1700, 2500)\), it is observed that \(y_{i}\) for \(i \in \mathcal{V}_{1}\) maintain tracking within a tiny boundary of 3, and \(y_{i}\) for \(i \in \mathcal{V}_{2}\) within a small boundary of -4, with respect to \(y_{d}(k)=1\). The consistent tracking within specified boundaries, with respect to \(y_{d}(k)\), demonstrates that the asymmetric bipartite objective, without an aperiodic DoS attack is achieved.

\subsection{DDET under Aperiodic DoS with Time-Varing Leader}

In this example, the nonlinear MAS are considered, with the models of 8 agents given to produce the nonlinear data flow as follows
\begin{equation*}
\begin{aligned}
&y_{1}(k+1) = \frac{y_{1}(k) u_{1}(k)}{1 + y_{1}(k)^{2}} + 2 u_{1}(k) \\
&y_{2}(k+1) = \frac{y_{2}(k) u_{2}(k)}{1 + y_{2}(k)^{3}} + 5 u_{2}(k) \\
&y_{3}(k+1) = \frac{y_{3}(k) u_{3}(k)}{1 + y_{3}(k)^{2}} + 3 u_{3}(k) \\
&y_{4}(k+1) = \frac{y_{4}(k) u_{4}(k)}{1 + y_{4}(k)^{3}} + 5 u_{4}(k) \\
&y_{5}(k+1) = \frac{y_{5}(k) u_{5}(k)}{1 + y_{5}(k)^{1}} + 0.8 u_{5}(k) \\
&y_{6}(k+1) = \frac{y_{6}(k) u_{6}(k)}{1 + y_{6}(k)^{0.9}} + 0.5 u_{6}(k) \\
&y_{7}(k+1) = \frac{y_{7}(k) u_{7}(k)}{1 + y_{7}(k)^{1.2}} + 0.4 u_{7}(k) \\
&y_{8}(k+1) = \frac{y_{8}(k) u_{8}(k)}{1 + y_{8}(k)^{2}} + 0.5 u_{8}(k)
\end{aligned}
\end{equation*}


The desired time-varing $y_d(k)$ for the MAS is set as
\begin{equation*}
\begin{aligned}
&y_d(k)=
\\&\left\{\begin{aligned}
&5\sin \left(\frac{\pi k}{2500}\right)+3\cos \left(\frac{\pi k}{2500}\right), &&k \in [0,1199]
\\&5\sin \left(\frac{\pi k}{5000}\right)+6\cos \left(\frac{\pi k}{6000}\right), &&k \in [1200, 2500]
\end{aligned}\right.
\end{aligned}    
\end{equation*}

\begin{figure}[ht]
\centering
\includegraphics[width=0.5\textwidth]{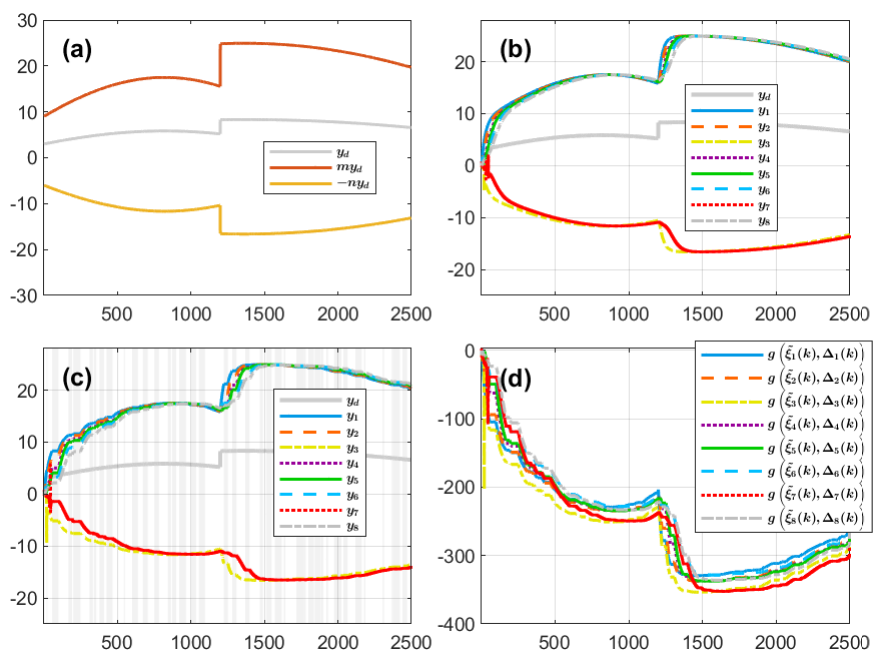}
\caption{(a) The grey line is reference leader $y_d$, the red line is the desired trajectory of $\mathcal{V}_1$, and the yellow line is the desired trajectory of $\mathcal{V}_2$. (b)  Output trajectories under the DDET protocols without aperiodic DoS attacks: The gray solid line represents the constant leader reference signal $y_d(k)$. (c) Output trajectories under the DDET protocols under aperiodic DoS attacks: The gray solid line represents the constant leader reference signal $y_d(k)$. (d) The event-triggered function value $g_i\left(\Delta_{i}(k), \tilde{e}_{y_i}(k)\right)$ based on the developed data-driven parameter while under aperiodic DoS attacks.}
\label{Figure_2}
\end{figure}

From Figs. \ref{Figure_1} and \ref{Figure_2}, we observe that followers maintain the desired trajectory within a tight boundary as dictated by the proposed DDET, illustrating its success in resolving the ABC problem for nonlinear MAS under aperiodic DoS attacks. These findings validate the proposed DDET's ability to facilitate ABC tracking under an event-triggered communication scheme, highlighting its resilience against aperiodic DoS attacks. The protocol's effectiveness, relying solely on agents' triggered outputs and inputs, showcases the data-driven essence of the control method. Therefore, the DDET not only exemplifies leader-follower unified ABC control but also confirms the boundedness of all signals in challenging communication environments.

\section{CONCLUSION}
\label{sec:con}
The ABC problem of networked nonlinear discrete-time MAS with both cooperative and antagonistic interactions under aperiodic DoS attacks has been investigated. The distributed DDET protocols, in the forms of state feedback and output feedback, are designed based on the predicted values of neighboring agents' states and their triggered outputs, respectively. Two distinct numerical examples with constant and time-varying reference values have demonstrated the effectiveness and resilience of the proposed DDET protocols.

\ifCLASSOPTIONcaptionsoff
  \newpage
\fi
\ifCLASSOPTIONcaptionsoff
  \newpage
\fi
\bibliographystyle{IEEEtran}
\bibliography{ref}

\end{document}